\documentclass[conference,letterpaper]{IEEEtran}

\IEEEoverridecommandlockouts

\usepackage{amsmath,amssymb,amsthm}
\usepackage{dsfont}
\usepackage[hidelinks,breaklinks=true]{hyperref}
\usepackage{tikz}
\usepackage{subcaption}
\usepackage{caption}
\usepackage{graphicx}
\usepackage{multicol}
\usepackage{multirow}
\usepackage{amsthm}
\usepackage{graphicx}
\usepackage{algorithm,algorithmic}
\usepackage{mathtools}
\usepackage{amsmath}
\usepackage{xfrac}
\usepackage{stmaryrd}
\usepackage{cryptocode}
\usepackage[framemethod=TikZ]{mdframed}
\usepackage[compress]{cite}

\usetikzlibrary{arrows,positioning,automata,shapes,calc,patterns}

\DeclareMathOperator*{\esssup}{ess\,sup}

\newcommand{\range}[1]{\llbracket #1 \rrbracket}

\newtheorem{theorem}{Theorem}

\newtheorem{definition}{Definition}

\newtheorem{proposition}{Proposition}

\begin{document}

%

\title{Noiseless Privacy } 

\author{
\IEEEauthorblockN{
Farhad~Farokhi}
\IEEEauthorblockA{
CSIRO's Data61 \\ 
The University of Melbourne \\
farhad.farokhi@data61.csiro.au \\
farhad.farokhi@unimelb.edu.au
}
}

\maketitle

\begin{abstract}
In this paper, we define noiseless privacy, as a non-stochastic rival to differential privacy, requiring that the outputs of a mechanism (i.e., function composition of a privacy-preserving mapping and a query) can attain only a few values while varying the data of an individual (the logarithm of the number of the distinct values is bounded by the privacy budget). Therefore, the output of the mechanism is not fully informative of the data of the individuals in the dataset. We prove several guarantees for noiselessly-private mechanisms. The information content of the output about the data of an individual, even if an adversary knows all the other entries of the private dataset, is bounded by the privacy budget. The zero-error capacity of memory-less channels using noiselessly private mechanisms for transmission is upper bounded by the privacy budget. The performance of a non-stochastic hypothesis-testing adversary is bounded again by the privacy budget. Finally, assuming that an adversary has access to a stochastic prior on the dataset, we prove that the estimation error of the adversary for individual entries of the dataset is lower bounded by a decreasing function of the privacy budget. In this case, we also show that the maximal information leakage is bounded by the privacy budget. In addition to privacy guarantees, we prove that noiselessly-private mechanisms admit composition theorem and  post-processing does not weaken their privacy guarantees. We prove that quantization operators can ensure noiseless privacy if the number of quantization levels is appropriately selected based on the sensitivity of the query and the privacy budget. Finally, we illustrate the privacy merits of noiseless privacy using multiple datasets in energy and transport.
\end{abstract}


\section{Introduction}
Big data revolution, equipped with novel tools for data collection, analysis, and reporting, has significant promises for answering societal challenges. These promises however come at the cost of erosion of privacy.  Therefore, there is a need for rigorous protection of the privacy of individuals.

Natural candidates for privacy protection, such as differential privacy~\cite{dwork2006calibrating,dwork2014algorithmic} and information-theoretic privacy~\cite{sankar2013utility,yamamoto1983source}, require randomized policies for privacy protection.
The definition of differential privacy assumes the use of randomized functions as well as the probability of outputs, and conventional information-theoretic tools, such as mutual information and entropy, rely on random variables. 

Heuristic-based privacy-preserving methods, such as $k$-anonymity~\cite{samarati2001protecting, sweeney2002k} and $\ell$-diversity~\cite{1617392}, are however deterministic in nature. They employ deterministic mechanisms, such as suppression and generalization, and do not assume stochastic properties about the datasets. Popularity of these methods is evident from the availability of toolboxes for implementation\footnote{\url{https://arx.deidentifier.org/overview/related-software/}}.

Although providing powerful guarantees, randomized or stochastic privacy-preserving policies sometimes cause problems, such as un-truthfulness~\cite{bild2018safepub}, that are undesirable in practice~\cite{Poulis2015}. This is touted as a reason behind slow adoption of differential privacy within financial and health sectors~\cite{bild2018safepub}. For instance, randomized policies, stemming from differential privacy in financial auditing, complicate fraud detection~\cite{bhaskar2011noiseless, nabar2006towards}. Randomized policies can also generate unreasonable and unrealistic outputs that might mislead investors or market operators, e.g., by reporting noisy outputs that point to lack of liquidity in a financial sector while that was not the case. For instance, the slow-decaying nature of the Laplace noise means impossible reports (e.g., negative median income) can occur with a non-negative probability~\cite{bambauer2013fool}. Randomized privacy-preserving policies have also encountered difficulties in medical, health, or social sciences~\cite{dankar2013practicing,Mervis114}. Furthermore, the Laplace mechanism, a common approach to ensuring differential privacy, is shown to cause undesirable properties, e.g., the optimal estimation in the presence of Laplace noise is computationally expensive~\cite{farokhi2016optimal}. 
These motivate the development of non-stochastic privacy metrics and privacy-preserving policies in a rigorous manner. 

Although it has been proved that noiseless policies cannot provide the strong guarantees of randomized policies, e.g., it has been proved that differential privacy cannot be delivered without noise~\cite{dwork2014algorithmic}, the popularity of noiseless privacy-preserving policies justifies investigating metrics for analysis and comparison. This must be done irrespective of their inherent philosophical weaknesses in comparison to stochastic policies because they belong to a different category. 

In this paper, we define the new notion of noiseless privacy. Noiseless privacy implies that the outputs of a mechanism can only attain a few distinct values while varying the data of an individual. Therefore, the output of the mechanism is not very informative about the data of the individuals in a dataset. We prove the following guarantees for the noiselessly-private mechanisms:
\begin{itemize}
\item The information content of the output about the data of each individual, even if an adversary knows all the other entries of the private dataset, is bounded from above by the privacy budget (a constant similar to the privacy budget in differential privacy capturing the amount of the leaked information). As non-stochastic notions of information, we use non-stochastic information leakage in~\cite{8662687} and the maximin information~\cite{nair2013nonstochastic}. These are established measures of information in the non-stochastic information theory literature~\cite{duan2014transfer,nair2013nonstochastic,8662687, lim2014deterministic}. 
\item Zero-error capacity of memory-less channels using noiselessly-private mechanisms for data transmission is upper bounded by the privacy budget. Zero-error capacity is the non-stochastic  equivalent of normal capacity, also coined by Shannon while investigating non-stochastic communication channels and worst-case behaviours~\cite{shannon1956zero}. 
\item The performance of an adversary performing non-stochastic hypothesis tests~\cite{farokhiCDC2019} on the data of an individual, while knowing all the other entries of the private dataset, is bounded again by the privacy budget. 
\item Assuming that an adversary has access to a stochastic prior about the dataset, we prove that the error of an adversary for estimating the data each individual is lower bounded by a decreasing function of the privacy budget. Therefore, by reducing the privacy budget, the estimation error of the adversary worsens. In this case, we also show that the maximal information leakage (in the sense of~\cite{issa2018operational}) is upper bounded by the privacy budget. Hence, by reducing the privacy budget, we can also reduce the maximal information leakage.
\end{itemize}
In addition to these privacy guarantees, we prove the following important properties:
\begin{itemize}
\item Noiselessly-private mechanisms admit composition theorem, i.e., the privacy budgets of the mechanisms add up when reporting on multiple queries on the same private dataset.
\item Post-processing of noiselessly-private mechanisms does not weaken their privacy guarantees, i.e., the privacy budget can only be increased by post-processing. 
\end{itemize}
We also prove that quantization operators can ensure noiseless privacy. We provide a recipe for determining the number of quantization levels based on the sensitivity of the query and the privacy budget. Finally, we illustrate the privacy merits of noiseless privacy using multiple datasets in energy and transport.

\subsection{Related Studies}
\paragraph*{Anonymization} Anonymization is widely used within public and private sectors for releasing sensitive datasets\footnote{See \url{https://data.gov.au} and \url{https://www.kaggle.com} for examples.} for public competitions and analysis. Although popularly adopted, anonymization is often insufficient for privacy preservation~\cite{narayanan2008robust, su2017anonymizing, de2013unique} and hence, systematic methods with provable guarantees are required.

\paragraph*{Multi-Party Computation and Encryption}
We may use secure multi-party computation, for instance based on homomorphic encryption, to compute aggregate statistics or machine learning models~\cite{popa2011privacy, li2010secure, Lindell1010073540445986_3,du2004privacy,vaidya2002privacy,vaidya2008privacy,jagannathan2005privacy}. Secure multi-party computation and homomorphic encryption introduce massive computation and communication overheads. They also do not fully eliminate the risk of privacy breaches, e.g., risks associated with dis-aggregation attacks still remain if these algorithms are not pared with other privacy-preserving techniques.

\paragraph*{Differential Privacy} Differential privacy offers provable privacy guarantees~\cite{dwork2008differential, dwork2014algorithmic, duchi2013local,kairouz2014extremal, machanavajjhala2008privacy, hall2012random, padakandla2018preserving}. This method uses randomization to provide plausible deniability for the data of an individual by ensuring that the statistics of privacy-preserving outputs do not change significantly by varying the data of individual. Additive Laplace and Gaussian noise with scales proportional to the sensitivity of the submitted query with respect to the individual entries of the dataset  are proved to guarantee differential privacy~\cite{dwork2014algorithmic}. By definition, differential privacy requires randomization.

\paragraph*{Information-Theoretic Privacy} Information-theoretic privacy, a rival to differential privacy, dates back to studying secrecy~\cite{6772207} and its generalizations~\cite{sankar2013utility,courtade2012information,yamamoto1983source, yamamoto1988rate}.  Information-theoretic guarantees have been also used to measure the quantity of leaked private information when using differential privacy~\cite{Aceto2011,du2012privacy}. In information-theoretic privacy, entropy, mutual information, Kulback-Leiber divergence, and Fisher information have been  repeatedly used as measures of privacy~\cite{farokhi2015quadratic,wainwright2012privacy, liang2009information, lai2011privacy,li2015privacy, bassi2018lossy, farokhi2018fisher}. Information theory, starting with Shannon~\cite{shannon1948mathematical}, assumes that data source and communication channels are random, and is powerful in modelling and analysing communication systems. However, traditional notions in information theory, such as mutual information, are not useful for analysing non-stochastic/noiseless settings and deterministic privacy-preserving policies.

\paragraph*{Deterministic Privacy-Preserving Policies}
Noiseless privacy-preserving policies are often heuristic-based making them vulnerable to attacks, e.g., $k$-anonymity is vulnerable to homogeneity attack~\cite{1617392}. This is because we do not possess sensible measures/definitions for privacy that extend to noiseless privacy-preserving policies on deterministic datasets. Therefore, we cannot prove, in any sense, privacy guarantees of noiseless privacy-preserving (even if weak or limited in scope or practice). The  popularity of noiseless privacy-preserving policies justifies investigating metrics for analysis and comparison. In this paper, we propose a rival to differential privacy that is noiseless. We use non-stochastic information theory, non-stochastic hypothesis testing, stochastic estimation theory to investigate the merits of this definition. Non-stochastic information theory dates back to early studies of information transmission~\cite{hartley1928transmission, kolmogorov1959varepsilon,renyi1961measures, nair2013nonstochastic,jagerman1969varepsilon}. It has been recently used in engineering~\cite{nair2012nonstochastic, duan2015transfer,wiese2016uncertain}. Most recently, non-stochastic information theory was used in~\cite{8662687} for investigating deterministic privacy-preserving policies. Interestingly, in~\cite{8662687}, it was easily proved that $k$-anonymity is not privacy-preserving using non-stochastic information theory, a fact that was only observed using adversarial attacks in~\cite{1617392}.  

\subsection{Paper Organization}
The rest of the paper is organized as follows. Background material on non-stochastic information theory and hypothesis testing are presented in Section~\ref{sec:background}. Noiseless privacy is defined in Section~\ref{sec:noiselessprivacy_def}. In this section, guarantees and properties of noiseless privacy are also presented. A method for ensuring noiseless privacy is presented in Section~\ref{sec:noiselessprivacy_satisfaction}. Experimental results are presented in Section~\ref{sec:experiments}. Finally, the paper is concluded in Section~\ref{sec:conclusions}.

\section{Uncertain Variables, Hypothesis Testing, and Non-Stochastic Information Theory} \label{sec:background}
We start by reviewing necessary concepts from non-stochastic information theory, particularly, uncertain variables, non-stochastic information leakage, and hypothesis testing.
\subsection{Uncertain  Variables}
Let  $\Omega$ be an uncertainty set/space whose elements, i.e.,   $\omega\in \Omega$, model/capture the source of uncertainty. An uncertain variable $X$ is defined as a mapping on $\Omega$. For any uncertain variable $X:\Omega\rightarrow\mathbb{X}$, $X(\omega)$ is the realization of uncertain variable $X$ (corresponding to the realization of uncertainty $\omega\in \Omega$). \textit{Marginal range} of uncertain variable $X$ is
$\range{X}:=\{X(\omega):\omega\in\Omega \}\subseteq
\mathbb{X}.$
\textit{Joint range} of uncertain variables $X:\Omega\rightarrow\mathbb{X}$
and $Y:\Omega\rightarrow\mathbb{Y}$ is defined as $
\range{X,Y}:=\{(X(\omega), Y(\omega)):
\omega\in\Omega \}\subseteq \mathbb{X}\times \mathbb{Y}.$
\textit{Conditional range} of uncertain variable $X$, conditioned on the realizations of uncertain variable $Y$ belonging to the set $\mathcal{Y}$, i.e., $Y(\omega)\in\mathcal{Y}\subseteq\range{Y}$, is given by 
$\range{X|\mathcal{Y}}:=\{X(\omega):\exists \omega\in\Omega  \mbox{ such that } Y(\omega)\in\mathcal{Y}\} \subseteq \range{X}.$
If $\mathcal{Y}$ is a singleton, i.e., $\mathcal{Y}=\{y\}$, we use $\range{X|y}$ instead of $\range{X|\{y\}}=\range{X|\mathcal{Y}}$. The definition of uncertain variables and their properties are similar to those of random variables with the exception of not requiring a measure on  $\Omega$. Finally, if the marginal range $\range{X}$ is uncountably infinite for an uncertain variable $X$, we refer to $X$ as a \textit{continuous} uncertain variable, similar to a continuous random variable. If the marginal range $\range{X}$ is countable for an uncertain variable $X$, we call $X$ a \textit{discrete} uncertain variable. 

\subsection{Non-Stochastic Information Theory}
Non-stochastic entropy of discrete uncertain variable $X$ is
\begin{align} \label{eqn:firstentropys}
H_0(X):=\log_2(|\range{X}|)\in \mathbb{R}\cup\{\pm \infty\}.
\end{align}
This is commonly referred to as the Hartley entropy~\cite{hartley1928transmission, nair2013nonstochastic}. For continuous uncertain variable $X$, the non-stochastic (differential) entropy is given by
\begin{align} \label{eqn:firstentropys1}
h_0(X):=\log_e(\mu(|\range{X}|))\in \mathbb{R}\cup\{\pm \infty\},
\end{align}
where $\mu(\cdot)$ is the Lebesgue measure. This is sometimes referred to as R\'{e}nyi differential 0-entropy~\cite{nair2013nonstochastic}.  The authors of~\cite{nair2013nonstochastic,shingin2012disturbance} define  the non-stochastic conditional entropy of uncertain variable $X$, conditioned on  uncertain variable $Y$, as
\begin{align}
H_0(X|Y):=\max_{y\in \range{Y} } \log_2(|\range{X|y}|),
\end{align}
for discrete uncertain variables $X$ and $Y$. Similarly, for continuous uncertain variables $X$ and $Y$, we get
\begin{align}
h_0(X|Y):=\esssup_{y\in \range{Y} } \log_e(\mu(\range{X|y})).
\end{align}
Now, we can define non-stochastic information between uncertain variables $X$ and $Y$ as the difference of the entropy of $X$ with and without access to realizations of $Y$. Hence, for discrete uncertain variables, non-stochastic information can be defined as
\begin{align}
I_0(X;Y):=&H_0(X)-H_0(X|Y)\nonumber\\
=&\min_{y\in \range{Y} } \log_2\left(\frac{|\range{X}|}{|\range{X|y}|} \right).\label{eqn:information1}
\end{align} 
For continuous uncertain variables, non-stochastic information can be similarly defined as $I_0(X;Y):=h_0(X)-h_0(X|Y)$. It is clear that the non-stochastic information is in fact not symmetric, i.e., $I_0(X;Y)\neq I_0(Y;X)$ in general. 

With slight adaptation, Kolmogorov had previously defined `combinatorial' conditional entropy using $\log(|\range{X|y}|)$ and the information gain  as $|\range{X}|/|\range{X|y}|$ in~\cite{kolmogorov1959varepsilon}. The combinatorial conditional entropy and the information gain are only defined for a fixed realization $Y(\omega) = y$ while~\eqref{eqn:information1} is based on the worst-case scenario. 

In~\cite{8662687}, it was observed that,  in the context of information-theoretic privacy, the non-stochastic information~\eqref{eqn:information1} is not a good measure of information leakage and therefore, the  non-stochastic information leakage was proposed as
\begin{align}
L_0(X;Y):=&\max_{y\in \range{Y} } \log_2\left(\frac{|\range{X}|}{|\range{X|y}|} \right),\label{eqn:information2}
\end{align}
for discrete uncertain variables. Similarly, for continuous uncertain variables, the non-stochastic information leakage was defined as 
\begin{align}
L_0(X;Y):=&\esssup_{y\in \range{Y} } \log_e\left(\frac{\mu(\range{X})}{\mu(\range{X|y})} \right).\label{eqn:information3}
\end{align}
In general, the non-stochastic information $I_0$ and non-stochastic information leakage $L_0$ are not equal, i.e,  $I_0(X;Y)\neq L_0(Y;X)$.
In fact, from the definition, it is easy to see that $I_0(X;Y)\leq L_0(X;Y)$. 
Further, $L_0(X;Y)$ is not symmetric. We propose the symmetrized non-stochastic information leakage as 
\begin{align}
L_0^{\mathrm{s}}(X;Y):=\min(L_0(X;Y),L_0(Y;X)).
\end{align}
Note that, by construction, $L_0^{\mathrm{s}}(X;Y)=L_0^{\mathrm{s}}(Y;X)$.

\paragraph{Maximin Information}
In~\cite{nair2013nonstochastic}, the maximin information was introduced as a symmetric measure of information and its relationship with zero-error capacity was explored. To present the definition of the maximin information, we need to introduce the notion of overlap partitions:
	\begin{itemize}
		\item $x,x'\in\range{X}$ are $\range{X|Y}$-overlap connected, or in short $x\leftrightsquigarrow x'$,  if there exists a finite sequence of conditional ranges $\{\range{X|y_i}\}_{i=1}^n$ such that $x\in\range{X|y_1}$, $x'\in\range{X|y_n}$, and $\range{X|y_i}\cap \range{X|y_{i+1}}\neq \emptyset$ for all $i=1,\dots,n-1$;
		\item $\mathcal{A}\subseteq\range{X}$ is $\range{X|Y}$-overlap connected if all $x,x'\in\mathcal{A}$ are $\range{X|Y}$-overlap connected;
		\item $\mathcal{A},\mathcal{B}\subseteq\range{X}$ are $\range{X|Y}$-overlap isolated if there does not exist $x\in\mathcal{A},x'\in\mathcal{B}$ such that $x,x'$ are $\range{X|Y}$-overlap connected;
		\item An $\range{X|Y}$-overlap partition is a partition of $\range{X}$ such that  each member set is $\range{X|Y}$-overlap connected and all two member sets are $\range{X|Y}$-overlap isolated.
	\end{itemize}

Symmetry, i.e., $x\leftrightsquigarrow x'$ implies that $x'\leftrightsquigarrow x$, and transitivity, i.e., $x\leftrightsquigarrow x'$ and $x'\leftrightsquigarrow x''$ implies that $x\leftrightsquigarrow x''$, guarantee that a unique $\range{X|Y}$-overlap partition always exists~\cite{nair2013nonstochastic}. The unique $\range{X|Y}$-overlap partition is shown by $\range{X|Y}_\star$ in what follows. The maximin information is
\begin{align}
I_\star(X;Y):=\log_2(|\range{X|Y}_\star|).
\end{align}
In~\cite{nair2013nonstochastic}, it was proved that  $|\range{X|Y}_\star|=|\range{Y|X}_\star|$ and thus $I_\star(X;Y)=I_\star(Y;X)$. We now prove an important result regarding the relationship between non-stochastic information leakage and maximin information. 

\begin{proposition} \label{prop:1}  For discrete uncertain variable $Y$, $I_\star(X;Y)\leq L_0^{\mathrm{s}}(X;Y)$.
\end{proposition}

\begin{proof} See Appendix~\ref{proof:prop:1}.
\end{proof}

An uncertain time series $X$ is a sequence of uncertain variables $X[k]:\Omega\rightarrow\mathbb{X}$ for all $k\in\mathbb{N}$. Alternatively, we can think of uncertain time series $X$ as a mapping from the sample space $\Omega$ to the set of discrete-time functions $\mathbb{X}^\infty:=\{\forall x:\mathbb{N}\rightarrow \mathbb{X}\}$. 

Now, we can define a memory-less uncertain communication channel. A memory-less uncertain channel maps any uncertain time series $X$ to uncertain time series $Y$ such that 
\begin{align*}
\range{&Y[k],\dots,Y[1]|X[k](\omega)=x[k],\dots,X[1](\omega)=x[1]}\\
&=\range{Y[k]|X[k](\omega)=x[k]}\times \cdots \times \range{Y[1]|X[1](\omega)=x[1]},
\end{align*}
for all $(x[k],\dots,x[1])\in\range{X[k],\dots,X[1]}$ and all $k\in\mathbb{N}$. 
A code of length $k$ is a finite set $\mathcal{F}\subseteq \mathbb{X}^k$ with each codeword $f\in\mathcal{F}$ denoting a distinct message. Define
\begin{align*}
\mathcal{X}&(y[k],\dots,y[1])\\
&:=\range{X[k],\dots,X[1]|Y[k](\omega)=y[k],\dots,Y[1](\omega)=y[1]}.
\end{align*}
The zero-error capacity is 
\begin{align}
C_0:=\lim_{k\rightarrow\infty}\sup_{
	\scriptsize
\begin{array}{c}
\mathcal{F}\subseteq \mathbb{X}^k: \\
|\mathcal{F}\cap \mathcal{X}(y[k],\dots,y[1])|\leq 1,\\
\forall (y[k],\dots,y[1])\in\mathbb{Y}^k
\end{array}
}
\frac{\log_2(|\mathcal{F}|)}{k}.
\end{align}
In what follows, we only consider sequence of discrete uncertain variables $Y[k]$.  Now, we are ready to relate the symmetrized non-stochastic information leakage to zero-error capacity of memory-less uncertain channels.

\begin{proposition} \label{prop:zeroerrorcapacity}
Any memory-less uncertain channel satisfies $
C_0\leq \sup_{\range{X[k]}\subseteq\mathbb{X}} L_0^{\mathrm{s}}(X[k];Y[k]).$
\end{proposition}

\begin{proof} See Appendix~\ref{proof:prop:zeroerrorcapacity}.
\end{proof}

\subsection{Non-Stochastic Hypothesis Testing}

\begin{figure}[t]
\centering
\begin{tikzpicture}
 \draw[fill=black!15]  plot[rotate=90,smooth,tension=.7,scale=0.3] coordinates {(-3.5,0.5) (-3,2.5) (-1,3.5) (1.5,3) (4,3.5) (5,2.5) (5,0.5) (2.5,-2) (0,-0.5) (-3,-2) (-3.5,0.5)};
\draw [fill=black!15,scale=.4,xshift=5cm] plot [rotate=50,smooth cycle] %
    coordinates {(-1.14,-1)(-0.84, -.18) (-0.04, 0.3) (2.24, 0) %
    (4.48, -0.56) (4.48, -3.46) (2.38,-4.84)(0.38, -1.28)};
\draw[-] (3.4,1.2)-- (3.06,-0.33);
\node[draw,rectangle,minimum width=1cm,minimum height=2cm,fill=black!15] at (6,0.2) {};
\draw[-] (5.5,0.2)-- (6.5,0.2);
\node[] at (6,0.7) {$p_0$};
\node[] at (6,-0.3) {$p_1$};
\draw[dashed,-] (-0.37,1.58) -- (3.4,1.22);
\draw[dashed,-] (+0.43,-1.02) -- (2.06,-0.6);
\path[dash dot,->,>=stealth'] (2.2,-.1) edge[bend right] (5.8,-.4);
\path[dash dot,->,>=stealth'] (3.6,.6) edge[bend left] (5.8,.8);
\node[] at (3.7,-0.1) {$\range{X|p_0}$};
\node[rotate=45] at (2.6,0.3) {$\range{X|p_1}$};
\node[] at (-.3,0.4) {$\Omega$};
\node[] at (3.8,1.4) {$\range{X}$};
\node[] at (6.0,1.5) {$\range{H}$};
\draw[pattern=north west lines]  plot[rotate=25,smooth,tension=.9,scale=0.25,yshift=-14.cm,xshift=10cm] coordinates {(-3.5,0.5) (-3,2.5) (-1,3.5) (1.5,3) (3,3.5) (4,2.5) (4,0.5) (2.5,-1) (1,-0.5) (-3,-1) (-3.5,0.5)};
\node[circle,draw,minimum width=1.5cm,pattern=north east lines] at (2.8,-2) {};
\path[dash dot,->,>=stealth'] (2.8,.1) edge[bend right] (2.6,-1.3);
\path[dash dot,->,>=stealth'] (3.8,-.4) edge[bend left] (3.6,-1.3);
\node[] at (5.0,-2.7) {\scriptsize $\range{Y|\range{X|p_0}}\setminus\range{Y|\range{X|p_1}}$};
\node[] at (0.3,-2) {\scriptsize $\range{Y|\range{X|p_1}}\setminus\range{Y|\range{X|p_0}}$};
\node[] at (2.2,-3.2) {\scriptsize $\range{Y|\range{X|p_1}}\cap \range{Y|\range{X|p_0}}$};
\path[<->,>=stealth'] (2.8,-3.0) edge[bend right] (3.2,-2.2);
\path[<->,>=stealth'] (4.2,-1.7) edge[bend left] (5.2,-2.5);
\path[<->,>=stealth'] (2.4,-2.3) edge[bend left] (1.2,-2.15);
\node[] at (4.9,-1.2) {$\range{Y}$};
\end{tikzpicture}
\caption{\label{fig:uncertainvariable} Relationship between uncertain variables in non-stochastic hypothesis testing based on uncertain measurements. If the realization of uncertain measurement $Y$ belongs to $\range{Y|p_0}\cap\range{Y|p_1}$, there is not enough evidence to accept or reject the null hypothesis $p_0$ or the alternative hypothesis $p_1$. However, if the realization of uncertain measurement  $Y$ belongs to $(\range{Y|p_0}\setminus\range{Y|p_1})$ ($(\range{Y|p_1}\setminus\range{Y|h_2})$), we can confidently accept (reject) the null hypothesis $p_0$ and reject (accept) the alternative hypothesis $p_1$.}
\end{figure}
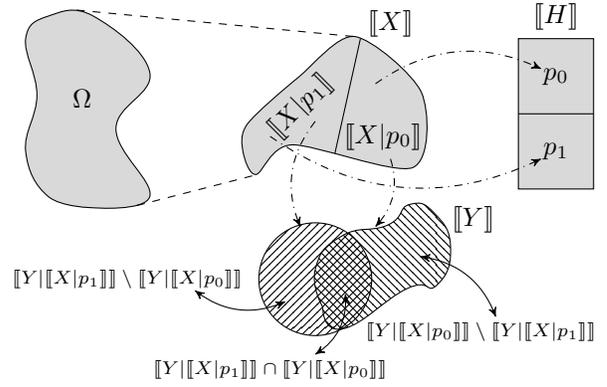

Consider uncertain variable $X$ denoting the original uncertain variable. An adversary is interested in testing the validity of a hypothesis for the realizations of $X$. The adversary does not have access to realizations of this uncertain variable as otherwise hypothesis testing is trivial. Instead, it has access to an \textit{uncertain measurement} of this variable denoted by $Y$. This is captured by that $Y(\omega)=g_Y(X(\omega))$ for a mapping $g_Y:\range{X}\rightarrow \range{Y}$. Recalling that uncertain variables are mappings from the uncertainty set, it must be that $Y=g_Y\circ X$, where $\circ$ denotes composition of mappings. Similarly, we may define the \textit{hypothesis} as an uncertain variable $H$ with binary range $\range{H}=\{p_0,p_1\}$, where $p_0$ denotes the \textit{null hypothesis} and $p_1$ denotes the \textit{alternative hypothesis}. We assume that there exists a mapping $g_H:\range{X}\rightarrow\range{H}$ such that  $H=g_H\circ X$; the hypothesis is constructed based on the uncertain variable $X$ as $H(\omega)=g_H(X(\omega))$. This setup and the relationship between all uncertain variables is summarized in Figure~\ref{fig:uncertainvariable}.

A \textit{test} is a function $T:\range{Y}\rightarrow\range{H}=\{p_0,p_1\}$. If $T(Y)=p_1$, the test rejects the null hypothesis in favour of the alternative hypothesis; however, if $T(Y)=p_0$, the test accepts the null hypothesis (and rejects the alternative hypothesis). The set of all tests is given by $\mathcal{C}(\range{H},\range{Y})$ which captures the set of all functions from $\range{Y}$ to $\range{H}$. Following~\cite{farokhiCDC2019}, we say that a test $T\in \mathcal{C}(\range{H},\range{Y})$ is correct at a particular realization of uncertain variable $Y$, $Y(\omega)=y\in\range{Y}$, if $\range{H|\range{X|y}}=\{T(y)\}$.  The set of all outputs at which test $T$ is correct is  equal to  $\aleph(T):=\{y\in\range{Y}:\range{H|\range{X|y}}=\{T(y)\}\}$. Based on this definition of correctness, we can define a performance measure for tests~\cite{farokhiCDC2019}. If $Y$ is a continuous uncertain variable, the performance is 
\begin{align} \label{eqn:performancetest_continuous}
\mathcal{P}(T):=\log_e(\mu(\aleph(T))).
\end{align}
Similarly, if $Y$ is a discrete uncertain variable, the performance is equal to
\begin{align}\label{eqn:performancetest_discrete}
\mathcal{P}(T):=\log_2(|\aleph(T)|).
\end{align}
In the following result, $\Delta$ denotes the symmetric difference operator on the sets, i.e., $\mathcal{A}\Delta\mathcal{B}=(\mathcal{A}\setminus\mathcal{B})\cup (\mathcal{B}\setminus\mathcal{A})$. 

\begin{proposition}[\hspace{-.0001in}\cite{farokhiCDC2019}] \label{tho:bound}  The performance of any test $T\in \mathcal{C}(\range{H},\range{Y})$ is  bounded by $ \mathcal{P}(T)\leq\log_e(\mu(\range{Y|p_0}\Delta\range{Y|p_1}))$
if $Y$ is a continuous uncertain variable, and by $\mathcal{P}(T)\leq
\log_2(|\range{Y|p_0}\Delta\range{Y|p_1}|)$ if $Y$ is a discrete uncertain variable.
\end{proposition}

Note that, for any realization of uncertain variable $Y$ in the set $\range{Y|p_0}\cap\range{Y|p_1}$, there is not enough evidence to accept or reject either the null hypothesis or the alternative hypothesis. This is because these realizations can be caused by realizations of $X$ that are consistent with the null hypothesis $p_0$ or realizations of $X$ that are consistent with the alternative hypothesis $p_1$. On the other hand, if the realization of the measurement  $Y$ is in the set $(\range{Y|p_0}\setminus\range{Y|p_1})\cup(\range{Y|p_1}\setminus\range{Y|h_2})=\range{Y|p_0}\Delta\range{Y|p_1}$, we can confidently reject or accept the null hypothesis or the alternative hypothesis. Proposition~\ref{tho:bound} can be thought of as a non-stochastic equivalent of the Chernoff-Stein Lemma; see, e.g.,~\cite[Ch.\,11]{cover2012elements} for randomized hypothesis testing. The size of the set  $\range{Y|p_0}\Delta\range{Y|p_1}$ essentially captures the difference between the ranges $\range{Y|p_0}$ and $\range{Y|p_1}$ resembling the Kullback--Leibler divergence in a non-stochastic framework.

\section{Noiseless Privacy: \\ Definition, Guarantees, and Properties}
\label{sec:noiselessprivacy_def}
We model a \textit{private dataset by a realization of a vector-valued uncertain variable} $X:\Omega\rightarrow\mathbb{R}^n$ with $n$ denoting the number of individuals whose data is in the dataset. The dataset is therefore in the form of
\begin{align*}
X(\omega)=
\begin{bmatrix}
X_1(\omega) \\
X_2(\omega) \\
\vdots \\
X_n(\omega)
\end{bmatrix},
\end{align*}
where $X_i(\omega)\in \mathbb{R}$ is the data of the $i$-th individual. Evidently, each $X_i:\Omega\rightarrow\mathbb{R}$, $1\leq i\leq n$, is itself an uncertain variable. 

A data curator, in possession of the realization of uncertain variable $X$, i.e., the private dataset $X(\omega)$, must return a response to a query
$f:\range{X}\rightarrow\mathbb{R}^q$ for some $q\in\mathbb{N}$. 
The curator employs a mechanism $\mathfrak{M}:\mathbb{R}^q\rightarrow\mathbb{R}^q$ to generate a a privacy-preserving response. Throughout this paper, $\mathfrak{M}\circ f$  is referred to as the mechanism of the curator. This is the same language used in the privacy literature albeit without the presence of the randomness~\cite{chatzikokolakis2013broadening, dwork2006calibrating}. Therefore, the curator provides the response $Y(\omega)=\mathfrak{M}\circ f(X(\omega))$. By definition, $Y=\mathfrak{M}\circ f\circ X$ is an uncertain variable. In the remainder of this paper, we use the notation $x_{-i}$ to denote $(x_1,\dots,x_{i-1},x_{i+1},\dots,x_n)$ for vectors and $X_{-i}$ to denote $(X_1,\dots,X_{i-1},X_{i+1},\dots,X_n)$ for uncertain variables alike. In this notation, $-i$ refers essentially refers to the set of all individuals except the $i$-th one. 

\begin{mdframed}[backgroundcolor=red!10,rightline=false,leftline=false,topline=false,bottomline=false,roundcorner=2mm] \vspace*{-.05in}
\begin{definition}[Noiseless Privacy] A mechanism $\mathfrak{M}\circ f$ is $\epsilon$-noiselessly private, for $\epsilon>0$, if 
	\begin{align}
	|\range{Y|X_{-i}(\omega)=x_{-i}}|\leq 2^\epsilon,\quad  \forall x_{-i}\in\range{X_{-i}},\forall i.
	\end{align}
\end{definition}
\end{mdframed}
%


Note that this definition is akin to a noiseless differential privacy. This is because, instead of bounding the information leakage as in information-theoretic privacy~\cite{8662687}, the output realizations are restricted if one individual entry of the dataset changes. Note that, when the data of $i$-th individual changes, the output can take all the values within the set  $\range{Y|X_{-i}(\omega)=x_{-i}}$. If this set is not informative, i.e., it does not contain many elements, reverse engineering the data of $i$-th individual changes with knowledge of $X(\omega)$ even in the presence of side-channel information is a difficult task. In what follows, we use non-stochastic information theory to establish the extend of the privacy guarantees from noiseless privacy. Similar to differential privacy, we can also define a local version of noiseless privacy. 

\begin{mdframed}[backgroundcolor=red!10,rightline=false,leftline=false,topline=false,bottomline=false,roundcorner=2mm] \vspace*{-.05in}
	\begin{definition}[Local Noiseless Privacy] Assume that $f_i:(x_i)_{i=1}^n\mapsto x_i$ for each $i$. A mechanism $\mathfrak{M}$ is $\epsilon$-locally noiselessly private if $\mathfrak{M}\circ f_i$ is $\epsilon$-noiselessly private for all $i$.
	\end{definition}
\end{mdframed}

\subsection{Guarantee: Non-Stochastic Information Leakage}
We define a function $\psi_{i,v_{-i}}:\range{X}\rightarrow \range{X_i}\times \{v_{-i}\}$ replacing the value of $X_j(\omega)$ or the realization of $X_j$, for all $1\leq j\leq$ except for $j=i$, with  given constants $v_j$, i.e., $\psi_{i,v_{-j}}(X(\omega))=(v_1, \dots,v_{i-1},X_{i}(\omega),v_{i+1},\dots,v_n).$ Let $\Psi_i = \{\psi_{i,v_{-j}} \,|\, v_{-j}\in \range{X_{-i}}\}$ be the set of all such functions for $i\in\{1,\dots,n\}$. The uncertain variable $\psi_{i,v_{-i}}\circ X$ becomes unrelated (in the sense of~\cite{nair2013nonstochastic}) to $X_{-i}$ for all $\psi_{i,v_{-i}}\in \Psi_i $. 

This definition allows us to measure the amount of the information that the curator's mechanism leaks about the data of the $i$-the individual $X_i(\omega)$. For a given $\psi_{i,v_{-i}}\in \Psi_i $, let us define  $Y=\mathfrak{M}\circ f\circ \psi_{i,v_{-i}} \circ X$. Now, the information between $Y$ and $X_i$ captures how much more information can an adversary extract from $Y$ knowing the data of all the individuals except the $i$-the individual. This is because, here, we let the adversary to select any possible $\psi_{i,v_{-i}}$. 

\begin{mdframed}[backgroundcolor=red!10,rightline=false,leftline=false,topline=false,bottomline=false,roundcorner=2mm] \vspace*{-.05in}
\begin{theorem}[Non-Stochastic Information vs Noiseless Privacy]
\label{tho:informationleakage_privacy} Assume $X$ is a discrete uncertain variable, $Y=\mathfrak{M}\circ f\circ \psi_{i,v_{-i}} \circ X$,  and $\mathfrak{M}\circ f$ is $\epsilon$-noiselessly private. For any  $\psi_{i,v_{-i}}\in \Psi_i $,   
\begin{align}
0\leq I_\star (X_i;Y)\leq   L_0^{\mathrm{s}}(X_i;Y)\leq L_0(Y;X_i)\leq \epsilon .
\end{align}
\end{theorem}
\end{mdframed}

\begin{proof} See Appendix~\ref{proof:tho:informationleakage_privacy}.
\end{proof}

Theorem~\ref{tho:informationleakage_privacy} shows that, by reducing $\epsilon$, we can reduce the amount of the leaked information about each individual. This makes sense. Consider the case where $\epsilon=+\infty$. In this case, the curator can report the output of the query $f(\psi_{i,v_{-i}} \circ X(\omega))$ completely (i.e., $\mathfrak{M}$ can be chosen to be equal identity) and the adversary, knowing $v_{-i}$, can compute the data of the data of the $i$-th individual $X_i(\omega)$ (at least if the adversary select the query to be linear with non-zero weight for the $i$-th individual). On the other hand, if $\epsilon=0$, the output becomes a constant that is independent of $X_{i}(\omega)$ and thus the adversary learns nothing new  about the data of the  $i$-th individual $X_i(\omega)$. 

\subsection{Guarantee: Zero-Error Capacity}

Let us consider a memory-less noiselessly-private communication channel. This can be seen as a non-stochastic equivalent of  differentially-private communication channels in~\cite{barthe2011information}. 

Let $\mathfrak{M}\circ f$ be a $\epsilon$-noiselessly private for some $\epsilon>0$.  For any given sequence of mappings $\{\psi^t_{i,v_{-i}}\}_{t\in\mathbb{N}}$ with $\psi^t_{i,v_{-i}}\in \Psi_i $, a memory-less $\epsilon$-noiselessly-private channel maps any uncertain time series $X=(X[k],\dots,X[1])$ to uncertain time series $Y=(Y[k],\dots,Y[1])$ such that $Y[\ell](\omega)=\mathfrak{M}\circ f\circ \psi^\ell_{i,v_{-i}} (X(\omega))$ for all $1\leq \ell \leq k$ and $k\in\mathbb{N}$.

This setup can be seen as a case in which the curator is reporting on a stream of data from the individuals. We can assume that an extremely strong adversary can set the realizations of the data of all individuals except the $i$-th individual. The capacity of the channel captures the amount of information that passes through a $\epsilon$-noiselessly private mechanism over time. 

\begin{mdframed}[backgroundcolor=red!10,rightline=false,leftline=false,topline=false,bottomline=false,roundcorner=2mm]  \vspace*{-.05in}
\begin{theorem}[Zero-Error Capacity vs Noiseless Privacy] Assume, for all $k$, $X[k]$ is a discrete uncertain variable, $Y[k]=\mathfrak{M}\circ f\circ \psi^\ell_{i,v_{-i}}\circ X[k]$, and $\mathfrak{M}\circ f$ is $\epsilon$-noiselessly private. For any $\psi_{i,v_{-i}}\in \Psi_i $, the zero-error capacity of memory-less $\epsilon$-noiselessly-private channel is bounded by 
	\begin{align}
C_0\leq \epsilon.
	\end{align}
\end{theorem}
\end{mdframed} 

\begin{proof}  The rest of the proof follows from  Proposition~\ref{prop:zeroerrorcapacity} and Theorem~\ref{tho:informationleakage_privacy}.
\end{proof}

%
%
%
%
%
%

\subsection{Guarantee: Non-Stochastic Hypothesis Testing}


In this part, our analysis is motivated by the definition of semantic security or  indistinguishability under chosen plaintext attack~\cite{katz2014introduction}. Assume that an adversary selects $i\in\{1,\dots,n\}$, $x^{}_i,x'_i\in\range{X_i}$, and provides this information to the curator. 
The curator uses uncertain variable $\overline{X}_i:\Omega\rightarrow\range{\overline{X}_i}:=\{x^{}_i,x'_i\}$ to constructs uncertain variable $\overline{X}=(X_{-i},\overline{X}_i)$.
Fix $\psi_{i,v_{-i}}\in\Psi_i$. The curator then generates a realization $\overline{X}(\omega)$, computes $Y(\omega)=\mathfrak{M}\circ f\circ\psi_{i,v_{-i}} (\overline{X}(\omega))$, and provides $Y(\omega)$ to the adversary. The adversary tests whether the realization of the data of individual $i$ is equal to $x^{}_i$ or $x'_i$ knowing that it is bound to be one of those values and knowing that the value of the data of all the other individuals is fixed to $v_{-i}$. We define the hypothesis uncertain variable $H$ using $g_H:\overline{X}(\omega)\mapsto H(\omega)$ as
\begin{align*}
H(\omega)=g_H(\overline{X}(\omega))=
\begin{cases}
p_0, & \overline{X}_i(\omega)=x^{}_i,\\
p_1, & \overline{X}_i(\omega)=x'_i.
\end{cases}
\end{align*}
The following theorem bounds the performance of the adversary for performing its hypothesis test. 

\begin{mdframed}[backgroundcolor=red!10,rightline=false,leftline=false,topline=false,bottomline=false,roundcorner=2mm]  \vspace*{-.05in}
\begin{theorem}[Hypothesis Testing vs Noiselss Privacy] \label{tho:nonstochastic_hypothesis_testing_privacy}
Assume $Y=\mathfrak{M}\circ f\circ\psi_{i,v_{-i}}\circ \bar{X}$ and $\mathfrak{M}\circ f$ is $\epsilon$-noiselessly private. Then, for any test $T$ and any  $\psi_{i,v_{-i}}\in \Psi_i $, the performance of the adversary is bounded by 
\begin{align}
\mathcal{P}(T)\leq \epsilon.
\end{align}
\end{theorem}
\end{mdframed}

\begin{proof}
See Appendix~\ref{proof:tho:nonstochastic_hypothesis_testing_privacy}.
\end{proof}

Bounding the performance of a hypothesis-testing adversary is in essence close to identifiability~\cite{wang2016relation,bkakria2018linking} for which privacy preservation relates to the potential of an adversary identifying the private data of individuals based on the received outputs.

\subsection{Guarantee: Performance of Adversaries with Stochastic Priors}
In this subsection, we briefly assume that the dataset is randomly distributed according to the probability density function $p$, i.e., for any Lebesgue-measurable set $\mathcal{X}\subseteq \range{X}$, $\mathbb{P}\{X\in  \mathcal{X}\}=\int_{x\in\range{X}}\xi(x)\mathds{d}\mu(x)$. We also consider an adversary that knows the realizations of all the entries of the dataset except the entry of the $i$-th individual. It constructs an estimate of the missing entry $X_i$ using an estimator $\hat{X}_i(X_{-i},\mathfrak{M}\circ f(X))$ using its prior information $X_{-i}$ and the response $\mathfrak{M}\circ f(X)$. 

\begin{mdframed}[backgroundcolor=red!10,rightline=false,leftline=false,topline=false,bottomline=false,roundcorner=2mm]  \vspace*{-.05in}
	\begin{theorem}[Stochastic Prior vs Noiselss Privacy]
		 \label{tho:stochastic_knowledge}
		 Assume that $\rho=\inf_{x\in\range{X}}\xi(x)>0$, $\mathfrak{M}\circ f$ is $\epsilon$-noiselessly private, and $f^{-1}\circ \mathfrak{M}^{-1}(y)$ is a connected set for any $y\in\range{Y}$. For any $p\in\mathbb{N}$,
		\begin{align}
	\mathbb{E}\{(X_i-\hat{X}_i(X_{-i},&\mathfrak{M}\circ f(X)))^p|X_{-i}\}\nonumber\\
	&\geq \left(\frac{\rho \mu(\range{X_i})^{p+1}}{2^{2p+2}}\right)2^{-\epsilon(p+1)}.
		\end{align}		
	\end{theorem}
\end{mdframed}

\begin{proof}
	See Appendix~\ref{proof:tho:stochastic_knowledge}.
\end{proof}

The lower bound on the adversary's estimation performance in Theorem~\ref{tho:stochastic_knowledge} is an decreasing function of $\epsilon$. Therefore, as expected and in-line with the earlier results, by decreasing $\epsilon$, we can reduce the adversary's ability to infringe on the privacy of any individual in the dataset even if the adversary knows the data of all the other individuals. 

\subsection{Guarantee: Stochastic Maximal Leakage}
In this section, we can recreate the stochastic framework for information leakage in~\cite{issa2018operational} by again endowing all the uncertain variables with a measure. This way, we can define the maximal stochastic leakage from $X$ to $Y$ as
\begin{align*}
\mathcal{L}_{\mathrm{c}}(X\rightarrow Y)=
\sup_{U-X-Y-\hat{U}} \log\left( \frac{\mathbb{P}\{U=\hat{U}\}}{\max_{u\in\range{U}}P_U(u)}\right)
\end{align*} 
where the supremum is taken over all random variables $U,\hat{U}$ taking values in the same finite arbitrary alphabets. Here, $U-X-Y-\hat{U}$ states that these variables from a Markov chain in the introduced order. It was shown in~\cite{issa2018operational} that 
\begin{align*}
\mathcal{L}_{\mathrm{c}}(X\rightarrow Y)
&=\log\left(\sum_{y\in\range{Y}}\max_{x\in\range{X}:P_{X}(x)>0}P_{Y|X}(y|x)\right)\\
&=I_{\infty} (X;Y).
\end{align*}

\begin{mdframed}[backgroundcolor=red!10,rightline=false,leftline=false,topline=false,bottomline=false,roundcorner=2mm]  \vspace*{-.05in}
	\begin{theorem}[Maximal Leakage vs Noiseless Privacy] Assume $Y=\mathfrak{M}\circ f\circ\psi_{i,v_{-i}}\circ X$ and $\mathfrak{M}\circ f$ is $\epsilon$-noiselessly private. Then,  $\mathcal{L}_{\mathrm{c}}(X_i\rightarrow Y)\leq \epsilon.$
	\end{theorem}
\end{mdframed}

\begin{IEEEproof} Note that $\sup_{P_{Y|X}}\mathcal{L}_{\mathrm{c}}(X\rightarrow Y)\leq H_0(Y)$ because of~\cite[Lemma~1 \& Example~6]{issa2018operational}. Furthermore, 
	$|\range{\mathfrak{M}\circ f\circ \psi_{i,v_{-i}} \circ X}|
	=|\range{Y|X_{-i}(\omega)=v_{-i}}|\leq 2^\epsilon.$
\end{IEEEproof}

Evidently, the amount of the leaked information is upper bounded by the privacy budget. Hence, by reducing the privacy budget, we can minimize the amount of the leaked information. 

\subsection{Property: Composition of Noiselessly-Private Mechanisms}
Composition of differentially-private mechanisms~\cite{dwork2014algorithmic, dwork2010boosting,kairouz2017composition} is an important result showing that the privacy budgets add up when reporting on multiple queries on the same private dataset. In what follows, we show that the same also applies to noiseless privacy.

\begin{mdframed}[backgroundcolor=red!10,rightline=false,leftline=false,topline=false,bottomline=false,roundcorner=2mm]  \vspace*{-.05in}
\begin{theorem}[Composition of Noiselessly-Private Mechanisms] \label{tho:composition_theorem} Let $\mathfrak{M}_1$ and $\mathfrak{M}_2$ be such that $\mathfrak{M}_1\circ f$ and $\mathfrak{M}_2\circ f$ are  $\epsilon_1$-noiseless private and $\epsilon_2$-noiseless private, respectively. Then, $(\mathfrak{M}_1,\mathfrak{M}_2)\circ f$ is $(\epsilon_1+\epsilon_2)$-noiseless private.
\end{theorem}
\end{mdframed}

\begin{proof} See Appendix~\ref{proof:tho:composition_theorem}.
\end{proof}

\subsection{Property: Post-Processing of Noiselessly-Private Mechanisms}
Finally, an important property of differentially-private mechanisms and information-theoretic privacy is that the privacy guarantees do not weaken by post-processing privacy-preserving outputs~\cite{dwork2014algorithmic}. In what follows, this also holds for noiselessly-private mechanisms as well.

\begin{mdframed}[backgroundcolor=red!10,rightline=false,leftline=false,topline=false,bottomline=false,roundcorner=2mm]  \vspace*{-.05in}
	\begin{theorem}[Post-Processing of Noiselessly-Private Mechanisms] \label{tho:postprocessing} Let $\mathfrak{M}$ be such that $\mathfrak{M}\circ f$ is $\epsilon$-noiseless private. Then, $g\circ \mathfrak{M}\circ f$ is also $\epsilon$-noiseless private for any mapping $g$.
	\end{theorem}
\end{mdframed}

\begin{proof}
See Appendix~\ref{proof:tho:postprocessing}.
\end{proof}

\section{Noiseless Privacy: Satisfaction}
\label{sec:noiselessprivacy_satisfaction}
We can ensure noiseless privacy using non-stochastic approaches, such as binning or quantization. To do so, first, we define linear quantizers. 

\begin{definition}[Linear Quantizer] A $q$-level quantizer $\mathcal{Q}:[x_{\min},x_{\max}]\rightarrow \{b_1,\dots,b_q\}$ is a piecewise constant function  defined as
\begin{align*}
Q(x)=
\begin{cases}
b_1, & x\in [x_1,x_2),\\
b_2, & x\in [x_2,x_3),\\
\hspace{.05in} \vdots & \hspace{.4in}\vdots\\
b_{q-1}, & x\in[x_{q-1},x_{q}),\\
b_q, & x\in[x_{q},x_{q+1}],
\end{cases}
\end{align*}
where $(b_i)_{i=1}^q$ are distinct symbols and $x_1\leq x_2\leq \cdots \leq x_q$ are real numbers such that $x_1=x_{\min}$, $x_{q+1}=x_{\max}$, $x_{i+1}-x_{i}=(x_{\max}-x_{\min})/q$ for all $1\leq i\leq q$. 
\end{definition}

We can show that linear quantizers can achieve noiseless privacy for any query on private datasets. This is proved in the next theorem.

\begin{mdframed}[backgroundcolor=red!10,rightline=false,leftline=false,topline=false,bottomline=false,roundcorner=2mm]  \vspace*{-.05in}
\begin{theorem} \label{tho:noiselessprivacyworks} Let $\range{f(X)|X}\subseteq[y_{\min},y_{\max}]$. Define sensitivity of query $f$ as
	\begin{align*}
		\mathcal{S}(f):=&\sup_{x_{-i}\in\range{X_{-i}}} \mu(f(\range{X_i}\times\{x_{-i}\}))\\
		=&\sup_{x_{-i}\in\range{X_{-i}}} \sup_{x^{}_i,x'_i\in\range{X_i}} |f(x'_i,x_{-i})-f(x^{}_i,x_{-i})|.
	\end{align*}
The mechanism $\mathcal{M}\circ f$ is $\epsilon$-noiseless private if $\mathcal{M}$ is a $q$-level quantizer with
\begin{align*}
q\leq \frac{2^\epsilon(y_{\max}-y_{\min})}{\mathcal{S}_f}.
\end{align*}
\end{theorem}
\end{mdframed}

\begin{proof}
See Appendix~\ref{proof:tho:noiselessprivacyworks}.
\end{proof}

\section{Experiments}
\label{sec:experiments}
\subsection{Energy Data: Reporting Aggregate}
For this part, we use a publicly available dataset from the  Ausgrid\footnote{\url{https://www.ausgrid.com.au/Industry/Innovation-and-research/Data-to-share/Solar-home-electricity-data}} containing half-hour smart meter measurements for 300 randomly-selected homes with rooftop solar systems over the period from 1 July 2010 to 30 June 2013. In this paper, we use the data over July 2012 to June 2013.

\begin{figure}
\begin{center}
\fbox{%
\pseudocode[colsep=-1em,addtolength=0em]{%
    \textbf{Adversary} \<\< \textbf{Curator} \\[][\hline]
    \text{\raisebox{-4pt}{select $i_0$ and $i_1$}} \<\sendmessageright*[1.1cm]{i_0 \text{ and } i_1}\< \raisebox{-4pt}{ $j\leftarrow \{0,1\}$}  \\
    \<\<  \text{select $i_2,\dots,i_n$} \\[-1em]
    \<\sendmessageleft*[1.1cm]{y,t}\<  \raisebox{-4pt}{$y=\mathfrak{M}\circ f(x_{i_j,t},x_{i_2,t},\dots,x_{i_n,t})$} \\[-1em]
    \raisebox{-4pt}{\quad estimate $j$}\< \sendmessageright*[1.1cm]{\hat{j}} \< \raisebox{-4pt}{\textbf{return} $j=\hat{j}$}
}
}
\end{center}
\caption{
\label{fig:game}
The timing of a game used for evaluating the ability of an adversary in guessing if the data of a particular individual belongs to a publicly-released noiselessly-private aggregate statistics.
}
\end{figure}

Let $x_{i,t}$ denote the consumption of house $i$ at day $t$. We consider reporting aggregate statistics
\begin{align*}
y_t=f((x_{i,t})_{i=1}^n)=
\frac{1}{n}(x_{1,t}+\cdots+x_{n,t}), \; \forall t.
\end{align*}
Here, $f$ denotes the query. 
We particularly use the mechanism in Theorem~\ref{tho:noiselessprivacyworks} to report noiselessly-private outputs. In this experiment, we test the ability of an adversary for inferring if a particular household has contributed to the aggregate or not as in~\cite{buescher2017two}. We use a game, as in~\cite{bohli2010privacy,buescher2017two}, to evaluate the ability of the adversary. The setup of the game is summarized in Figure~\ref{fig:game}. At first, the adversary can select two households $i_0,i_1$. The curator select one of those households uniformly at random $i_j$. It also selects an additional $n-1$ households. Then it reports the privacy-preserving aggregate outputs. Based on the reported output, the adversary guesses the participating household $i_{\hat{j}}$. The adversary's success or advantage is then defined as 
\begin{align*}
\mathrm{Adv}:=2|\mathbb{P}\{j=\hat{j}\}-1/2|. 
\end{align*}
Small $\mathrm{Adv}$ means that the adversary is as successful as randomly guessing and large $\mathrm{Adv}$ implies that the adversary is successful in recognizing the household participating in the aggregate. 

Similar to~\cite{buescher2017two}, we use three adversary policies. The first one is based on correlation. In this case, the adversary selects $j\in\{0,1\}$ based on the correlation between $(x_{i_j,t})_{t}$ and $(y_t)_{t}$. The second policy is based on mean square error. In this case, the adversary selects $j\in\{0,1\}$ by minimizing the square error $\|(x_{i_j,t})_{t}-(y_t)_{t}\|_2$. Finally, the last policy uses the relative peaks of each load profile $(x_{i_0,t})_{t}$, $(x_{i_1,t})_{t}$, and $(y_t)_{t}$. In this case, the adversary selects $j\in\{0,1\}$ based on the most common peaks between $(x_{i_0,t})_{t}$ and $(y_t)_{t}$, or $(x_{i_1,t})_{t}$ and $(y_t)_{t}$.

\begin{figure*}
	\centering
	\begin{tikzpicture}
	\node[] at (0,0) {\includegraphics[width=.35\linewidth]{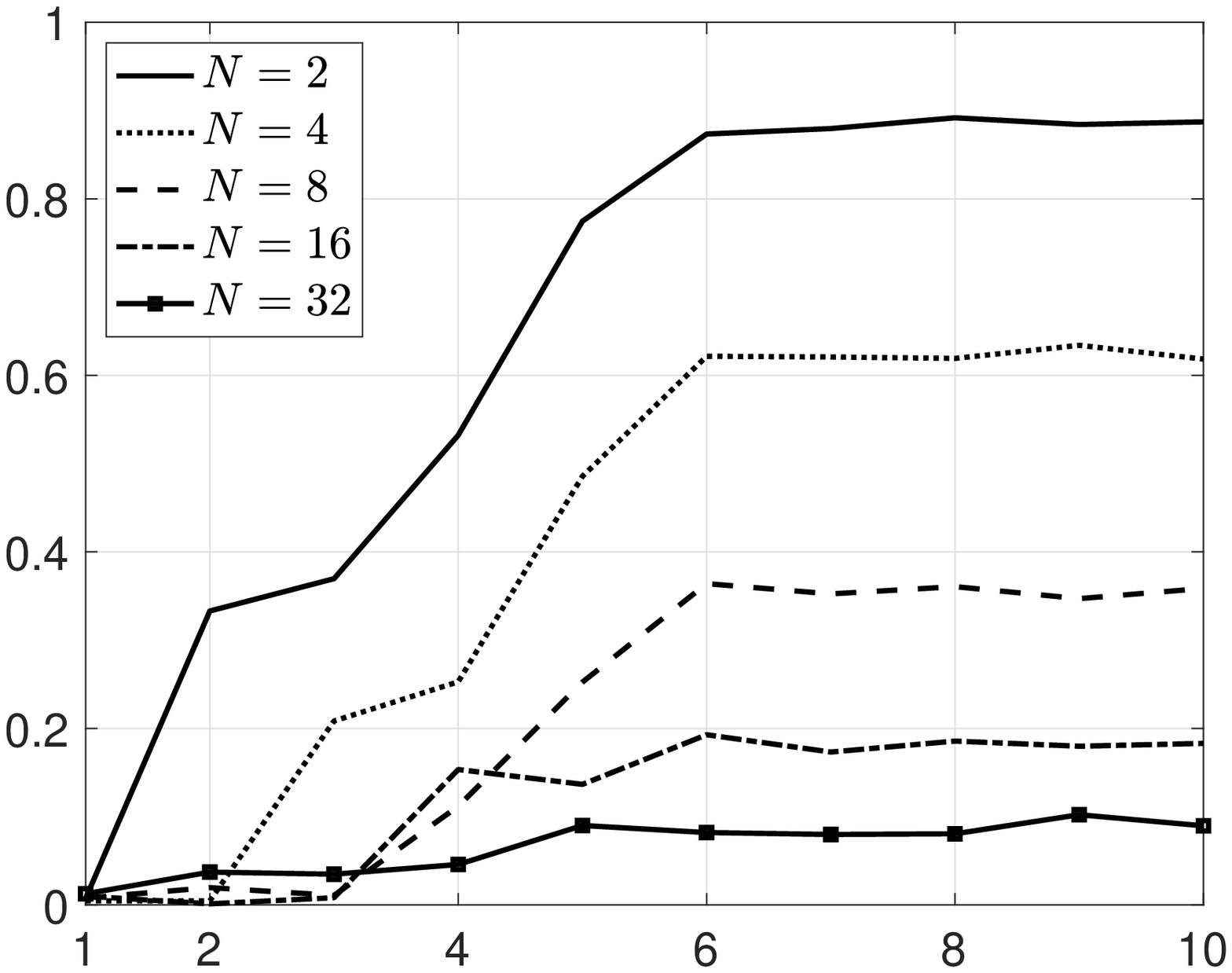}};
	\node[] at (6,0) {\includegraphics[width=.35\linewidth]{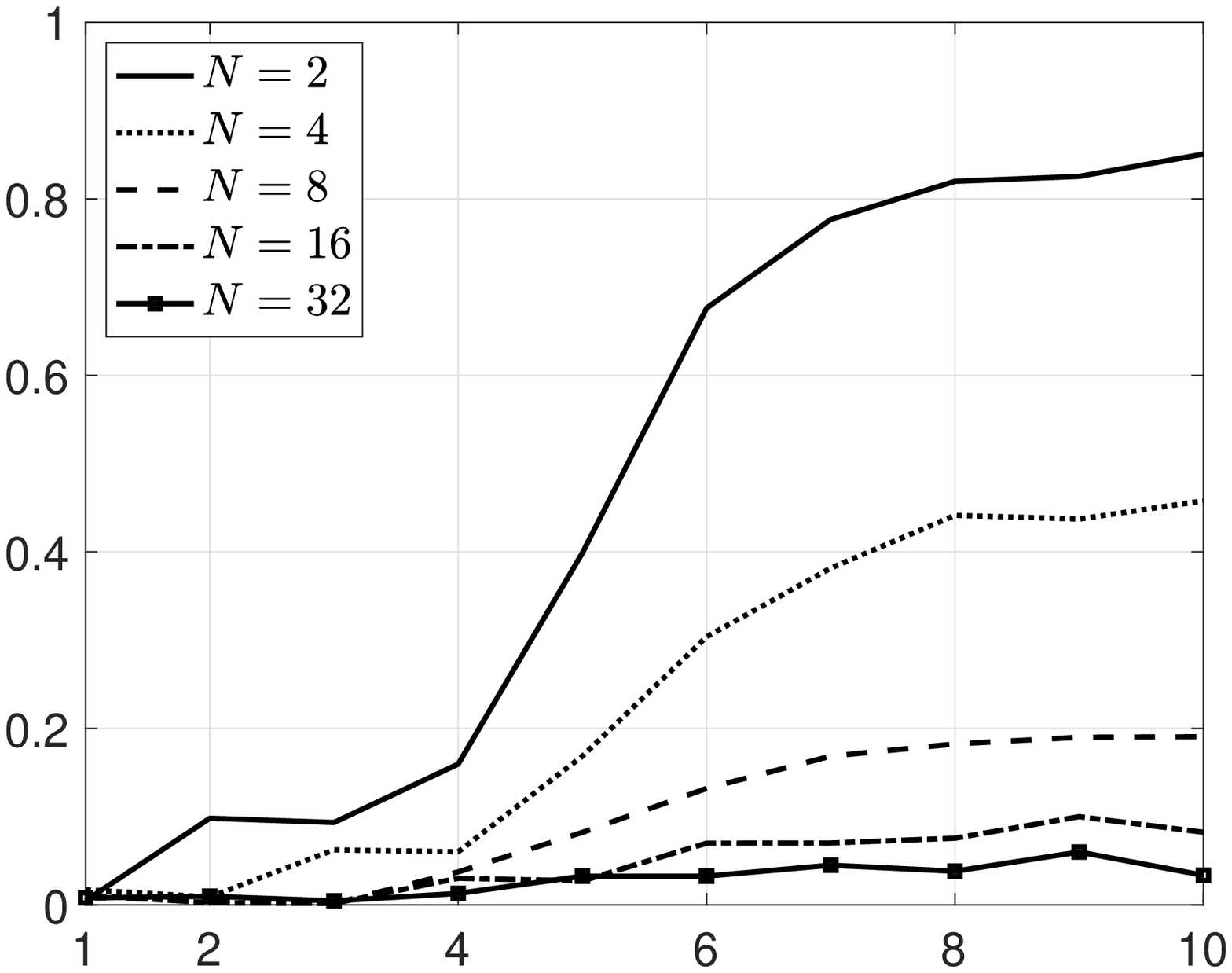}};	
	\node[] at (12,0) {\includegraphics[width=.35\linewidth]{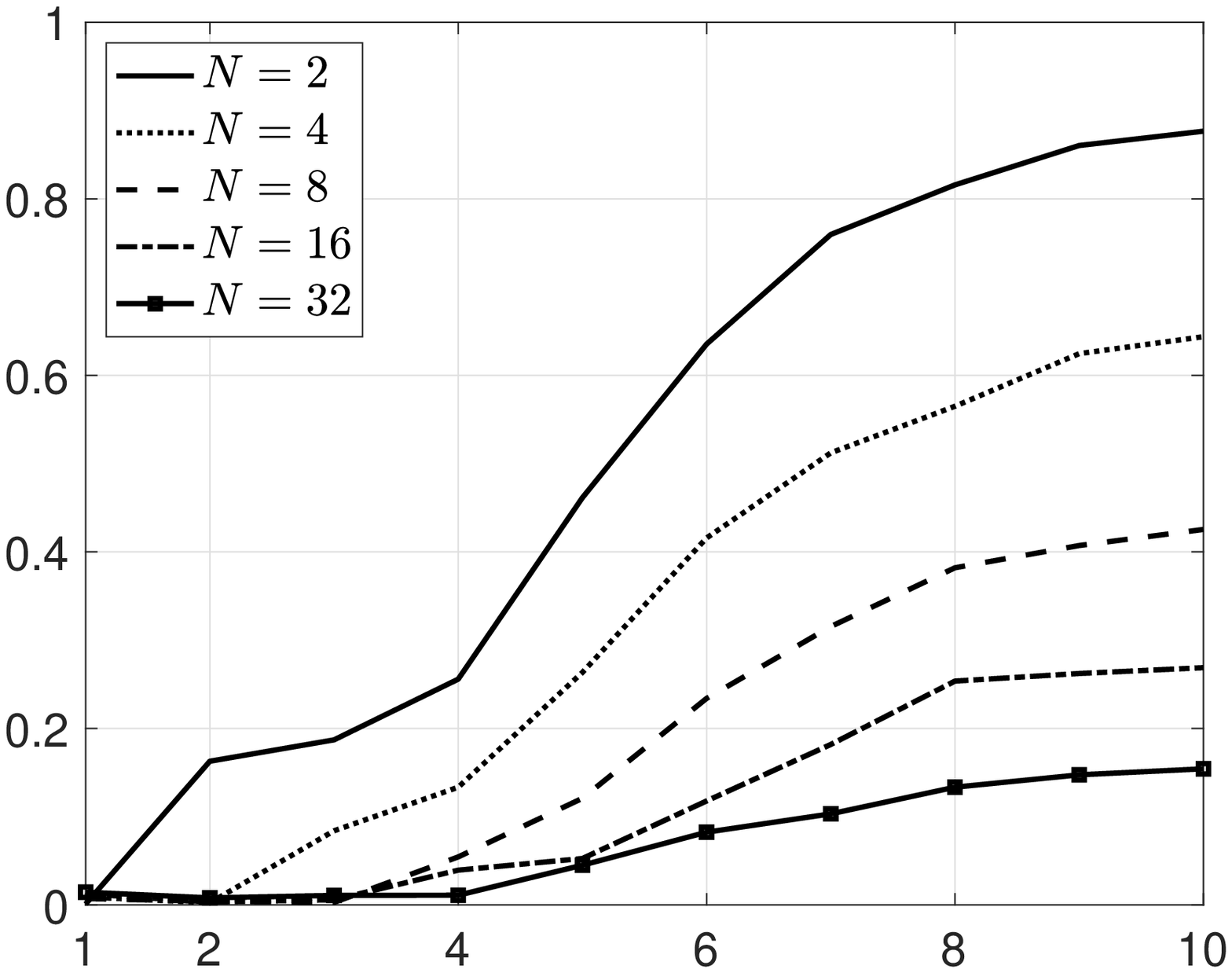}};	
	\node[rotate=90] at (-3,0) {$\mathrm{Adv}$};
	\node[] at (0,-2.3) {$\epsilon$};
	\node[rotate=90] at (+3,0) {$\mathrm{Adv}$};
	\node[] at (6,-2.3) {$\epsilon$};
	\node[rotate=90] at (+9,0) {$\mathrm{Adv}$};
	\node[] at (12,-2.3) {$\epsilon$};	
	\end{tikzpicture}
	\caption{
	\label{fig:adversary_game_aggregate}
	The advantage of the adversary $\mathrm{Adv}$ when using the correlation-based policy (left), the mean square error policy (center), and the peak-based policy (right). 
	}
\end{figure*}

Figure~\ref{fig:adversary_game_aggregate} illustrates the advantage of the adversary $\mathrm{Adv}$ when using the correlation-based policy (left), the mean square error policy (center), and the peak-based policy (right). As $\epsilon$ grows larger, the adversary's advantage tends toward the non-private case in~\cite{buescher2017two}. Clearly, even for moderate $\epsilon$ when considering small groups, the adversary's advantage  is very small. This is not the case for non-private outputs as observed in~\cite{buescher2017two}.
For instance, for small groups and moderate privacy budgets, such as $n=4$ and $\epsilon=2$ or $n=8$ and $\epsilon=3$, the adversary's advantage is negligible (almost zero). This shows that combining noiseless privacy with aggregation is an excellent tool for providing privacy to individuals.

\subsection{Energy Data: Reporting Single Consumption}
Non-intrusive load monitoring provides tools for extracting appliance-specific energy consumption statistics from the smart meter readings of a household and and is one of privacy concerns behind releasing energy data~\cite{zoha2012non,parson2012non}. In this section, we use Theorem~\ref{tho:noiselessprivacyworks} to report high-frequency energy consumption of a household in a privacy-preserving manner using local noiseless privacy. We then proceed to see the effect of privacy budget on an adversary performing non-intrusive load monitoring.

We use the low frequency data from the first house in the REDD database\footnote{\url{http://redd.csail.mit.edu/}} database~\cite{kolter2011redd}, which contains  the consumption of various appliances in the house every 3-4 seconds. This data in conjunction with the consumption of the entire house is used for training and verification of a non-intrusive load monitoring algorithm.  The consumption of the entire house is measured every second. The data is for the period of April 23--May 21, 2011. The part of the data prior to April 30th is used for training and the rest for validation purposes. We select the top 5  appliances in energy consumption for disaggregation purposes, namely, fridge, microwave, socket (in the kitchen), light, and dish washer. For non-intrusive load monitoring, we have used the NILMTK\footnote{\url{https://github.com/nilmtk/nilmtk}} toolbox in Python~\cite{NILMTK2014}. We have used a frequently utilized combinatorial optimization method for non-intrusive load monitoring. We report the success of the non-intrusive load monitoring using the $f$-score.

\begin{figure}
	\centering
	\begin{tikzpicture}
	\node[] at (0,0) {\includegraphics[width=1\linewidth]{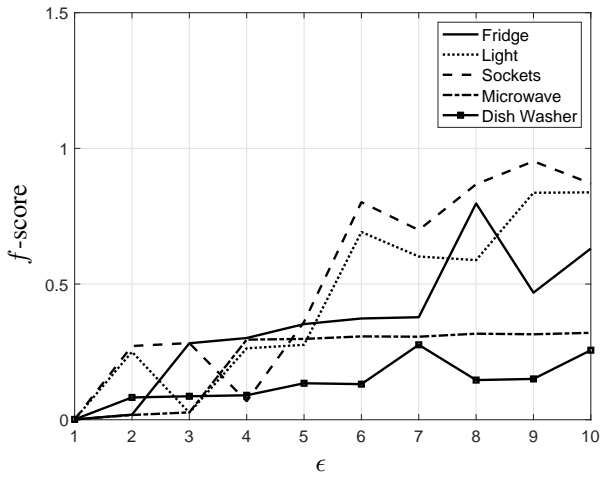}};
	\node[rotate=90] at (-4,0) {$f$-score};
	\node[] at (0,-3.2) {$\epsilon$};
	\end{tikzpicture}
	\caption{\label{fig:f-score} $f$-score of the non-intrusive load monitoring algorithm based on combinatorial optimization versus the privacy budget.}
\end{figure}

Figure~\ref{fig:f-score} shows the f-score of the non-intrusive load monitoring algorithm based on combinatorial optimization versus the privacy budget. As we can see the $f$-score gets rapidly bad as $\epsilon$ decrease. This means an adversary would not be able to identify the appliances that are used within the household reliably.  This illustrates the power of local noiseless privacy in reporting energy consumption of households for analysis while protecting the privacy of the households.

\begin{figure}
	\centering
\input{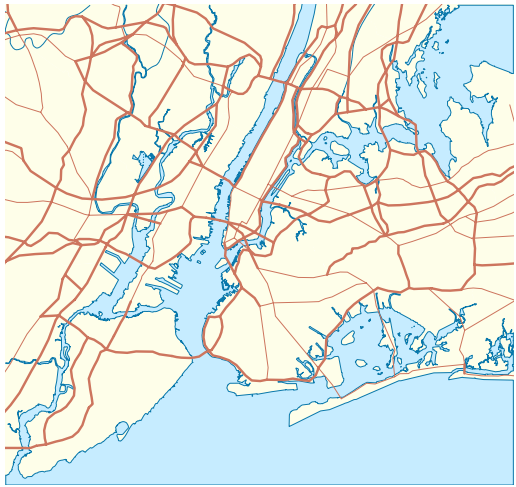}
\caption{\label{fig:NYCity} Map of New York City: latitude in [40.92$^\circ$, 40.49$^\circ$] and longitude in [-74.27$^\circ$,-73.68$^\circ$]}
\end{figure}

\begin{figure}
	\centering
\begin{tikzpicture}[>=stealth]
\node[] at (0,0) {\includegraphics[width=1\linewidth]{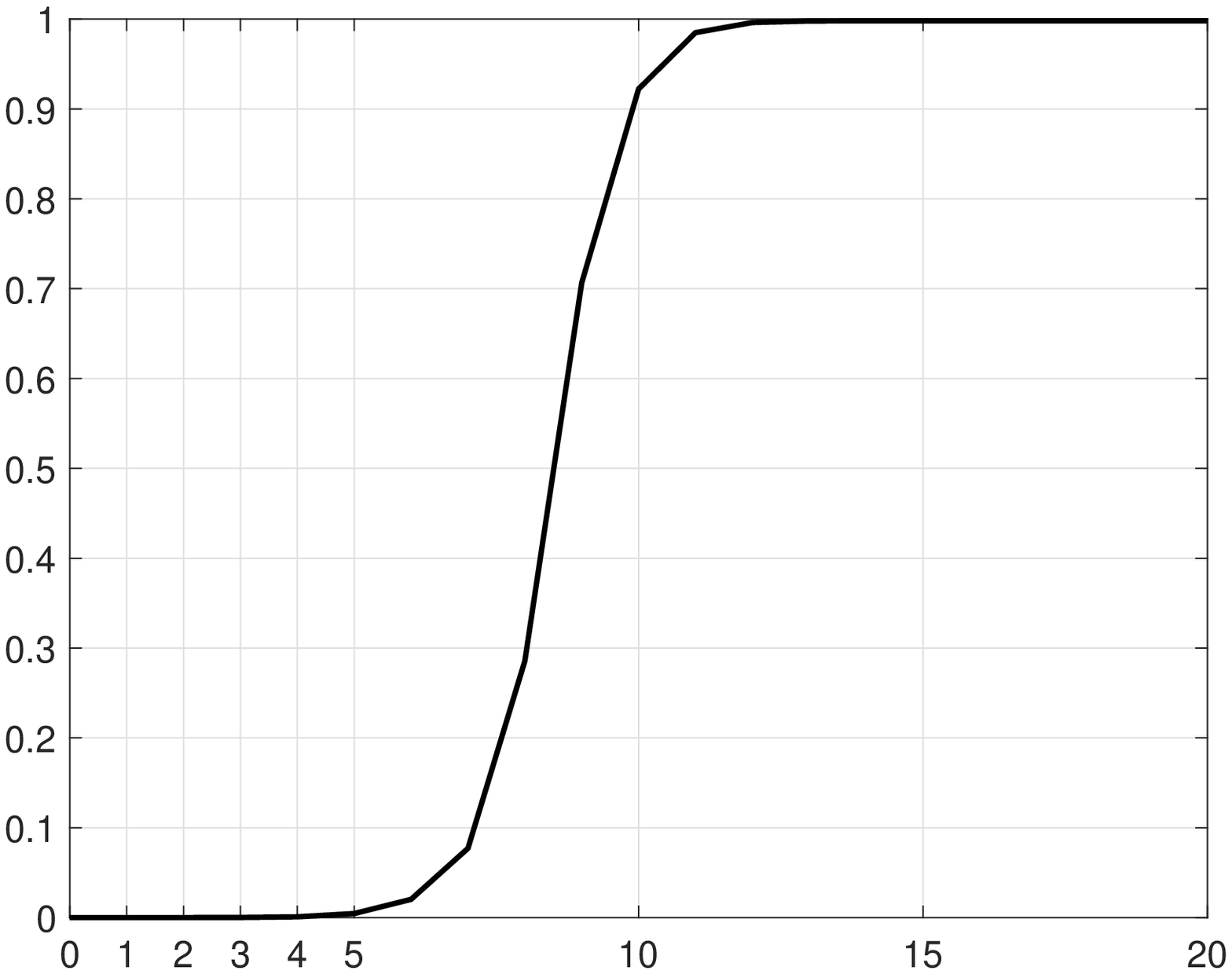}};
\node[] (1) at (-2.4,-1) {\small 
\begin{minipage}{2cm}
	\centering 
no unique \\ source-destinations
\end{minipage}
};
\draw[->] (1) -- (-3.22,-2.55);
\draw[->] (1) -- (-2.93,-2.55);
\node[] (2) at (2.5,1.8) {\small 
	\begin{minipage}{2cm}
	\centering 
	all unique \\ source-destinations
	\end{minipage}
};
\draw[->] (2) -- (+2.75,+2.79);
\node[] (3) at (-1.3,0.8) {\small 
	\begin{minipage}{2cm}
	\centering 
	19 unique \\ source-destinations
	\end{minipage}
};
\draw[->]  (3) -- (-1.3,-1.5)  --  (-2.55,-2.55);
\node[rotate=90] at (-4,0) {ratio of unique source-destinations};
\node[] at (0,-3.2) {$\epsilon$};
\end{tikzpicture}
\caption{
\label{fig:protion} Portion of taxi rides with unique start-end points versus the privacy budget.
}
\end{figure}
 

\subsection{Transport Data: Reporting Individual  Source-Destinations}
Finally, we use New York City Taxi Cab trips\footnote{\url{https://www.kaggle.com/kentonnlp/2014-new-york-city-taxi-trips}} for 2014. We use the first million trips and focus on trips that begin and end within the New York City in Figure~\ref{fig:NYCity}. Here, we consider reporting the start-end point of the taxi rides in a locally noiselessly private manner. In this subsection, we again use Theorem~\ref{tho:noiselessprivacyworks} to report start-end points of taxi rides in the New York City in a privacy-preserving manner using local noiseless privacy. In this case, for each $\epsilon$, we split the latitude and the longitude into $2^{\epsilon/2}$ boxes. Therefore, the privacy of the total privacy budget for the reported outputs is $2\epsilon$, following Theorem~\ref{tho:composition_theorem}.

Figure~\ref{fig:protion} illustrates the portion of unique start-end points of taxi rides versus the privacy budget. As we can see, the portion of unique start-end points is negligible for small $\epsilon$. This means an adversary would not be able to attribute a specific taxi ride to an individual, thus protecting the privacy of contributing individuals.

\section{Discussions and Future Work}
\label{sec:conclusions}
In this paper, we defined noiseless privacy, as a non-stochastic rival to differential privacy, requiring that the outputs of the mechanism to attain very few values while varying the data of an individual remains. We proved that noiseless-private mechanisms admit composition theorem and  post-processing does not weaken their privacy guarantees. We proved that quantization operators can ensure noiseless privacy. We finally illustrated the privacy merits of noiseless privacy and local noiseless privacy using multiple datasets in energy and transport.

\bibliographystyle{ieeetr}
\bibliography{scibib}

\appendices

\section{Proof of Proposition~\ref{prop:1}}
\label{proof:prop:1}
First, we prove that $I_\star(X;Y)\leq L_0(X;Y)$. 
Let $I_\star(X;Y)=m$. Then,  $\range{X|Y}_\star=\{P_1,\dots,P_{2^m}\}$. Each $P_i$ is non-empty. Therefore, there exists at least one $x$ such that  $x\in P_i$. Note that $x$ must also belong to $\range{X|y}$ for any $y\in\range{Y|x}$. We prove that $\range{X|y}\subseteq P_i$.  Assume that this not the case. Therefore, there exist an element of $x'\in\range{X|y}$, distinct from $x$, that belongs to another $P_j$, $j\neq i$, because $\{P_1,\dots,P_{2^m}\}$ covers $\range{X}$. We know that $P_i$ and $P_j$ are $\range{X|Y}$-overlap isolated by the definition of partition $\range{X|Y}_\star$. On the other hand, we evidently have $x\leftrightsquigarrow x'$ (by the definition of $\range{X|Y}$-overlap connectedness). This is a contradiction and thus $\range{X|y}$ must be a subset of $P_i$. This results in $|\range{X|y}|\leq |P_i|$ and hence
\begin{align}
\min_{y\in\range{Y}}|\range{X|y}|\leq P_i, \quad \forall i\in\{1,\dots,2^m\}. \label{eqn:proof:1}
\end{align}
On the other hand, $\bigcup_{i=1}^{2^m}P_i=\range{X}$ because $\{P_1,\dots,P_{2^m}\}$ is a partition for $\range{X}$. Because of the non-overlapping nature of the sets $\{P_1,\dots,P_{2^m}\}$, we get 
\begin{align}
\sum_{i=1}^{2^m} |P_i|=|\range{X}|.\label{eqn:proof:2}
\end{align}
Combining~\eqref{eqn:proof:1} and~\eqref{eqn:proof:2} results in $2^m\min_{y\in\range{Y}}|\range{X|y}|\leq |\range{X}|$. This implies that $I_\star(X;Y)\leq L_0(X;Y)$. Similarly, we can show that $I_\star(Y;X)\leq L_0(Y;X)$. By symmetry of the maximin information~\cite{nair2013nonstochastic}, i.e., $I_\star(X;Y)=I_\star(Y;X)$, we get that $I_\star(X;Y)\leq L_0(X;Y)$ and $I_\star(X;Y)\leq L_0(Y;X)$. This concludes the proof. 

\section{Proof of Proposition~\ref{prop:zeroerrorcapacity}}
\label{proof:prop:zeroerrorcapacity}
Let $\aleph(y[k],\dots,y[1])$ denote the statement  $Y[k](\omega)=y[k],\dots,Y[1](\omega)=y[1]$. Note that
\begin{align*}
\hspace{.4in}&\hspace{-.4in}2^{L_0(X[k],\dots,X[1];Y[k],\dots,Y[1])}\\&=\max_{(y[i])_{i=1}^k\in \range{Y}^k} \frac{|\range{X[k],\dots,X[1]}|}{|\range{X[k],\dots,X[1]|\aleph(y[k],\dots,y[1])}|} \nonumber
\\&= \frac{|\range{X[k],\dots,X[1]}|}{\min_{(y[i])_{i=1}^k\in \range{Y}^k}|\range{X_k,\dots,X_1|\aleph(y[k],\dots,y[1])}|} \nonumber
\\&= \frac{\prod_{\ell=1}^k |\range{X[\ell]}|}{\prod_{\ell=1}^k\min_{y[\ell]\in \range{Y}}|\range{X[\ell]|Y[\ell](\omega)=y[\ell]}|} \nonumber
\\&=\prod_{\ell=1}^k \frac{ |\range{X[\ell]}|}{\min_{y[\ell]\in \range{Y}}|\range{X[\ell]|Y[\ell](\omega)=y[\ell]}|} \\
&=\prod_{\ell=1}^k L_0(X[\ell];Y[\ell])\\
&=(L_0(X[\ell];Y[\ell]))^k. 
\end{align*}
Therefore, 
\[L_0(X[k],\dots,X[1];Y[k],\dots,Y[1])=kL_0(X[\ell];Y[\ell]),\]
and, as a result,
\[L_0(X[k],\dots,X[1];Y[k],\dots,Y[1])/k=L_0(X[\ell];Y[\ell]).\]
Similarly, we can show that 
\[L_0(Y[k],\dots,Y[1];X[k],\dots,X[1])/k=L_0(Y[\ell];X[\ell]).\]
Combining these inequalities with Proposition~\ref{prop:1} in this paper and Theorem 4.1 in~\cite{nair2013nonstochastic} proves the result. 

\section{Proof of Proposition~\ref{tho:informationleakage_privacy}}
\label{proof:tho:informationleakage_privacy}
Note that
\begin{align*}
2^{L_0^{\mathrm{s}}(X_i;Y)}
\leq & 2^{L_0(Y;X_i)} \\
=& \sup_{x_i\in\range{X_i}} \frac{|\range{Y}|}{\range{Y|X_i(\omega)=x_i}}\\
=& |\range{Y}|\\
=& \left|\bigcup_{x'_i\in\range{X_i}}\range{Y|X_i(\omega)=x'_i,X_{-i}(\omega)=v_{-i}}\right|\\
=& \left|\range{Y|X_{-i}(\omega)=v_{-i}}\right|\\
\leq& 2^\epsilon,
\end{align*}
where the second equality follows from that the realization of $Y$ can be uniquely determined based on the realization of $X_i$, i.e., $\range{Y|X_i(\omega)=x_i}$ is a singleton. Therefore, $L_0^{\mathrm{s}}(X_i;Y)\leq L_0(Y;X_i)\leq \epsilon$. The rest follows from Proposition~\ref{prop:1}.


\section{Proof of Theorem~\ref{tho:nonstochastic_hypothesis_testing_privacy}}
\label{proof:tho:nonstochastic_hypothesis_testing_privacy}
Since $Y$ is a discrete uncertain variable, for any test $T$, Theorem~\ref{tho:bound} states that the performance of the adversary is bounded from the above by $\mathcal{P}(T) \leq \log_2(|\range{Y|X_i(\omega)=x_i}\Delta\range{Y|X_i(\omega)=x'_i}|).$ Now, note that $\range{Y|X_i(\omega)=x_i}\Delta\range{Y|X_i(\omega)=x'_i}
\subseteq \range{Y|X_{-i}(\omega)=v_{-i}}. $
Therefore, $|\range{Y|X_i(\omega)=x_i}\Delta\range{Y|X_i(\omega)=x'_i}|\leq |\range{Y|X_{-i}(\omega)=v_{-i}}|\leq 2^\epsilon$. This concludes the proof.

\section{Proof of Theorem~\ref{tho:stochastic_knowledge}}
\label{proof:tho:stochastic_knowledge}
Define function $g(\cdot)$ such that $g(X_i)=\mathfrak{M}\circ f(X)$. 
There must exists $y\in\range{Y}$ such that $\mu(g^{-1}(y)))\geq \mu(\range{X_i})2^{-\epsilon}$. As otherwise, $\mu(g^{-1}(y))< \mu(\range{X_i})2^{-\epsilon}$ for all $y\in\range{Y}$  and thus
\begin{align*}
\mu(\range{X_i})
&=
\mu\left(\bigcup_{y\in\range{Y}}g^{-1}(y)\right)\\
&=\sum_{y\in\range{Y}}\mu\left(g^{-1}(y)\right)\\
&<\sum_{y\in\range{Y}}\mu(\range{X_i})2^{-\epsilon}\\
&=\mu(\range{X_i}).
\end{align*}
This is a contradiction. Hence, we get
\begin{align*}
\mathbb{E}&\{(X_i-\hat{X}_i(X_{-i},\mathfrak{M}\circ f(X))^p|X_{-i}\}\\
&\geq \rho \int_{g^{-1}(y)} (X_i-\hat{X}_i(X_{-i},\mathfrak{M}\circ f(X))^p \mathrm{d}\mu(X_i),
\end{align*}
where $\rho=\inf_{X\in\range{X}} \xi(X)$. 
Since $g^{-1}(y)$ is a connected set, there must exists $\underline{x}_i,\overline{x}_i$ such that closure of the $g^{-1}(y)$ is equal to $[\underline{x}_i,\overline{x}_i]$. Hence, we get 
\begin{align*}
\int_{g^{-1}(y)} (X_i-&\hat{X}_i(X_{-i},\mathfrak{M}\circ f(X))^p \mathrm{d}\mu(X_i)
\\& =\int_{\underline{x}_i}^{\overline{x}_i} (X_i-\hat{X}_i(X_{-i},\mathfrak{M}\circ f(X))^p \mathrm{d}\mu(X_i)
\\&\geq \int_{\underline{x}_i}^{\overline{x}_i} (z-(\underline{x}_i+\overline{x}_i)/2)^p\mathrm{d}z\\
&\geq 2\left(\frac{\overline{x}_i-\underline{x}_i}{4}\right)^p \frac{\overline{x}_i-\underline{x}_i}{4}\\
&=\frac{\mu(g^{-1}(y))^{p+1}}{4^{p+1/2}}\\
&\geq \mu(\range{X_i})^{p+1}2^{-\epsilon(p+1)}/2^{2p+1}.
\end{align*}
This concludes the proof.
\section{Proof of Theorem~\ref{tho:composition_theorem}}
\label{proof:tho:composition_theorem}
Note that
\begin{align*}
\range{(\mathfrak{M}_1,\mathfrak{M}_2)&\circ f(X)|X_{-i}(\omega)=x_{-i}}\\
=&\range{(\mathfrak{M}_1\circ f(X),\mathfrak{M}_2\circ f(X))|X_{-i}(\omega)=x_{-i}}\\
\subseteq&\range{\mathfrak{M}_1\circ f(X)|X_{-i}(\omega)=x_{-i}}\\
&\times \range{\mathfrak{M}_2\circ f(X)|X_{-i}(\omega)=x_{-i}},
\end{align*}
and as a result 
\begin{align*}
\mu(\range{(\mathfrak{M}_1,\mathfrak{M}_2)&\circ f(X)|X_{-i}(\omega)=x_{-i}})\\
\leq &\mu(\range{\mathfrak{M}_1\circ f(X)|X_{-i}(\omega)=x_{-i}})\\
&\times \mu(\range{\mathfrak{M}_2\circ f(X)|X_{-i}(\omega)=x_{-i}}).
\end{align*}
Hence, 
\begin{align*}
\log_e(\mu(\range{(\mathfrak{M}_1,\mathfrak{M}_2)&\circ f(X)|X_{-i}(\omega)=x_{-i}}))
\\
&\leq \log_e(\mu(\range{\mathfrak{M}_1\circ f(X)|X_{-i}(\omega)=x_{-i}})\\
&\times\mu(\range{\mathfrak{M}_2\circ f(X)|X_{-i}(\omega)=x_{-i}}))\\
=&\log_e(\mu(\range{\mathfrak{M}_1\circ f(X)|X_{-i}(\omega)=x_{-i}}))\\
&+\log_e(\mu(\range{\mathfrak{M}_2\circ f(X)|X_{-i}(\omega)=x_{-i}})).
\end{align*}
The proof for discrete uncertain variables follow the same approach. 

\section{Proof of Theorem~\ref{tho:postprocessing}}
\label{proof:tho:postprocessing}
The proof follows from that $
|\range{g\circ \mathfrak{M}\circ f(X)|X_{-i}(\omega)=x_{-i}}|
\leq |\range{\mathfrak{M}\circ f(X)|X_{-i}(\omega)=x_{-i}}|$ for any $x_{-i}$. 

\section{Proof of Theorem~\ref{tho:noiselessprivacyworks}}
\label{proof:tho:noiselessprivacyworks}
For any given $x_{-i}\in\range{X_{-i}}$, due to continuity of $f$, we know that $f(\range{X_i}\times\{x_{-i}\})=f\circ \psi_{i,x_{-i}}(\range{X_{i}})\subseteq [y_{\min},y_{\max}]$ is a connected set (because $\range{X_{i}}$ is connected). Therefore, if $\mathfrak{M}$ is a $q$-level quantizer, $\range{Y|X_{-i}(\omega)=x_{-i}}=\mathfrak{M}\circ f(\range{X_i}\times\{x_{-i}\})$ can at most contain $q\mu(f(\range{X_i}\times\{x_{-i}\}))/(y_{\max}-y_{\min})$ points. Therefore, $|\range{Y|X_{-i}(\omega)=x_{-i}}|\leq q\mathcal{S}_f/(y_{\max}-y_{\min})$. This concludes the proof.

\end{document}